\documentclass[a4paper,USenglish,cleveref, autoref, thm-restate]{lipics-v2021}

\usepackage{booktabs}
\newcolumntype{Y}{>{\centering\arraybackslash}X}
\usepackage{bm}

\title{Shortest Distances as Enumeration Problem}

\author{Katrin Casel}{Hasso Plattner Institute, University of Potsdam, Potsdam, Germany \and \url{https://hpi.de/friedrich/people/katrin-casel.html}}{Katrin.Casel@hpi.de}{https://orcid.org/0000-0001-6146-8684}{}
\author{Tobias Friedrich}{Hasso Plattner Institute, University of Potsdam, Potsdam, Germany \and \url{https://hpi.de/friedrich/people/tobias-friedrich.html}}{Tobias.Friedrich@hpi.de}{https://orcid.org/0000-0003-0076-6308}{}
\author{Stefan Neubert}{Hasso Plattner Institute, University of Potsdam, Potsdam, Germany \and \url{https://hpi.de/friedrich/people/stefan-neubert.html}}{Stefan.Neubert@hpi.de}{https://orcid.org/0000-0001-9148-6592}{}
\author{Markus L. Schmid}{Humboldt University, Berlin, Germany \and \url{http://www.mlschmid.de/}}{MLSchmid@MLSchmid.de}{https://orcid.org/0000-0001-5137-1504}{}

\authorrunning{K. Casel, T. Friedrich, S. Neubert and M. L. Schmid}

\Copyright{Katrin Casel, Tobias Friedrich, Stefan Neubert and Markus L. Schmid}



\category{} 

\relatedversion{} 

\nolinenumbers 

\hideLIPIcs  

\EventEditors{John Q. Open and Joan R. Access}
\EventNoEds{2}
\EventLongTitle{42nd Conference on Very Important Topics (CVIT 2016)}
\EventShortTitle{CVIT 2016}
\EventAcronym{CVIT}
\EventYear{2016}
\EventDate{December 24--27, 2016}
\EventLocation{Little Whinging, United Kingdom}
\EventLogo{}
\SeriesVolume{42}
\ArticleNo{23}

\newcommand{\maxdeg}{\Delta}
\newcommand{\avgdeg}{\overline{\Delta}}
\newcommand{\dist}{d}

\theoremstyle{definition}
\newtheorem*{problem}{Problem Definition}

\graphicspath{{figures/}}

\begin{document}

\maketitle

\begin{abstract}
We investigate the single source shortest distance (SSSD) and all pairs shortest distance (APSD) problems as enumeration problems (on unweighted and integer weighted graphs), meaning that the elements $(u, v, \dist(u, v))$ -- where $u$ and $v$ are vertices with shortest distance $\dist(u, v)$ -- are produced and listed one by one without repetition.
The performance is measured in the RAM model of computation with respect to preprocessing time and delay, i.\,e., the maximum time that elapses between two consecutive outputs.
This point of view reveals that specific types of output (e.\,g., excluding the non-reachable pairs $(u, v, \infty)$, or excluding the self-distances $(u, u, 0)$) and the order of enumeration (e.\,g., sorted by distance, sorted row-wise with respect to the distance matrix) have a huge impact on the complexity of APSD while they appear to have no effect on SSSD.

In particular, we show for APSD that enumeration without output restrictions is possible with delay in the order of the average degree.
Excluding non-reachable pairs, or requesting the output to be sorted by distance, increases this delay to the order of the maximum degree.
Further, for weighted graphs, a delay in the order of the average degree is also not possible without preprocessing or considering self-distances as output.
In contrast, for SSSD we find that a delay in the order of the maximum degree without preprocessing is attainable \emph{and} unavoidable for any of these requirements.
\end{abstract}

\section{Introduction}
\label{sec:introduction}
Computing the distance (and a corresponding shortest path) for each two vertices of a given (directed or undirected, weighted or unweighted) graph is a  fundamental computational problem.
Its theoretical research spans over more than half a century and many of the main results are nowadays part of a typical undergraduate computer science curriculum (see~\cite{cormenIntroductionAlgorithms2009}).
Nevertheless, even for such well-researched problems, a slight change in perspective can lead to new algorithmic challenges and reveal new interesting aspects.
As an example, the setting where the underlying graph is subject to constant change (e.\,g. insertions or deletions of edges) and the resulting task of maintaining the shortest distance information without recomputing it from scratch has given rise to an important and well-established research area within the field of dynamic graph algorithms (see, e.\,g., the surveys~\cite{DemetrescuItaliano2006, Martin2017arxiv}).

Our paradigm shift with respect to computing shortest distances lies in treating it as enumeration problem:
we want to enumerate all \emph{distance triples} $(u, v, \dist(u, v))$ (where $u$ and $v$ are vertices and $\dist(u, v)$ their distance) one by one, without repetitions.
The performance of algorithms solving this task is measured with respect to \emph{delay} (i.\,e., worst-case time between any two outputs, where the end-of-computation is also seen as an output) and \emph{preprocessing time} (i.\,e., additional time needed before enumeration starts).
While this general algorithmic framework has practical motivation and is heavily applied in other areas (see discussion of related work further below), we primarily use it in order to reveal new complexity aspects of shortest distance problems.
More specifically, we consider the problems \emph{single source shortest distance} (SSSD) and \emph{all pairs shortest distance} (APSD).

Many  algorithms for SSSD or APSD work by using some already computed distances in order to obtain new ones.
So, although not explicitly formulated as enumeration algorithms, such strategies nevertheless internally perform some kind of enumeration in the sense that at any point of the computation some distances are already known, while others are yet to be determined.
It is well known that the order of such hidden enumeration for shortest distances can have a crucial impact.
In particular, Dijkstra's algorithm fixes distances in increasing order and therefore is, in the comparison-model, limited by the lower bounds on sorting.

Explicitly treating SSSD and APSD as enumeration problems provides a fine-grained tool to reveal new and interesting aspects of computing shortest distances; thus, allowing for a taxonomy of shortest distance algorithms according to their enumeration behavior.
For example, our results show that even in the unweighted case the maximum degree of the graph is an upper and (unconditional) lower bound for the delay for SSSD, even if no enumeration order is required.
This bound consequently also holds for APSD if the distances are required to be enumerated ``row-wise'' (partially ordered by their first vertex).
However, if we drop this requirement for APSD, we can achieve algorithms with a delay in the order of the average degree of the graph.
While we explicitly chose to enumerate distances instead of paths, all our algorithms can be altered, without changing their performance, to also give pointers to the second vertex on a shortest path to enable recovering shortest paths from the output.

\subsection{Related Work}

In the following, $n$ denotes the number of vertices and $m$ denotes the number of edges of the input graph.

Most algorithms for SSSD on integer-weighted graphs, although they aim for a best possible total time, internally fix distances one by one.
The famous Dijkstra algorithm~\cite{dijkstraNoteTwoProblems1959}, and all its many subsequent improvements (see~\cite{older_survey} for a survey), fix distances one by one, in increasing order.
A sequence of papers achieved significant improvements by explicitly breaking from this order constraint:
Thorup developed a linear time (i.e., $O(m+n)$) algorithm for SSSD on undirected graphs with integer weights in~\cite{thorupUndirectedSinglesourceShortest1999}, and later generalized this to an almost linear time algorithm for directed graphs in~\cite{loglog}.
These ideas were later also used to develop algorithms for APSD which beat $n$ runs of Dijkstra's algorithm even for real-valued weights by Pettie and Ramachandran~\cite{real_apsp_undir} for undirected, and Pettie~\cite{apsp_reals} for directed graphs.
In a sense, these results have been achieved by changing the enumeration order.
In our enumeration point of view (which is, to the best of our knowledge, a new perspective w.r.t. shortest distance problems), we therefore also put a focus on the role of orders.

For APSD, there seem to be more algorithmic approaches that do not fix distances one at a time.
Probably the first thing that comes to mind in this regard is the classical Floyd-Warshall algorithm~\cite{floydAlgorithm97Shortest1962,warshallTheoremBooleanMatrices1962}, followed by the algorithms based on fast matrix multiplication by Seidel~\cite{seidelAllPairsShortestPathProblemUnweighted1995} and Alon~et~al.~\cite{alonExponentAllPairs1997}.
But also other approaches are more successful with non-trivial steps before fixing distances.
Chan~\cite{fast_apsp}, for example, shows how to solve APSD in unweighted, undirected graphs in $o(nm)$ by first heavily transforming the instance.
Hagerup~\cite{fast_dir_sssp} also uses a non-trivial structure he calls component tree.
For APSD enumeration, such strategies, at least at first glance, result in large preprocessing time.
We will however see a curious case where fast matrix multiplication may be used to improve the delay.

On the side of lower bounds, we like to mention the tight $\Theta(nm)$ bound on comparison-based algorithms for APSD with integer weights by Karger~et~al.~\cite{kargerFindingHiddenPath1993}.
Further, the lack of progress for APSD with general weights has lead to assuming the impossibility of significant improvement (i.\,e., a truly subcubic algorithm) in the form of the \emph{APSP conjecture}~\cite{apsp_home}.

Finally, let us discuss some aspects of the general enumeration framework adopted here.
In the literature, the term \emph{enumeration} is also used for the task of listing \emph{all} solutions to an instance of a combinatorial (optimization) problem, e.\,g., listing all maximal independent sets in a graph~\cite{is_enumeration}.
Our understanding of the term enumeration is rather ``enumerating \emph{the} solution'' instead of ``enumerating \emph{all} solutions''.
This perspective is mainly motivated by practical scenarios where a user is constantly \emph{querying} solutions for rapidly changing problem instances and where the need for seeing the (potentially large) solution in its entirety is negligible compared to the need for repeatedly accessing small portions of it very fast.

This also explains why in query evaluation in database theory this enumeration perspective is quite common and, over roughly the last decade, has even become a standard algorithmic setting, see~\cite{BerkholzEtAl2020, CaselSchmid21_conf, Segoufin2015} for more details.
Another interesting observation is that despite extensive work on enumeration problems in database theory, lower bounds on the delay are rather scarce in this area.
In fact, the few conditional lower bounds that do exist follow directly from the fact that some problems with lower bound assumptions can be encoded as getting the first element only, or computing the whole solution set; see~\cite[Section $6$]{BerkholzEtAl2020}.
In this regard, it seems noteworthy that in our study of shortest distances as enumeration problems, we are able to show an unconditional lower bound.
Note, however, that in the database theory setting we usually assume a linear time preprocessing phase as a prerequisite for the enumeration algorithms, while in this work we are primarily interested in enumeration algorithms without preprocessing and our lower bound is tailored to this setting.
Further lowering the delay by making use of a preprocessing phase is a secondary objective.

\subsection{Our Contribution}\label{sec:contribution}
In this paper we analyze enumeration of APSD for graphs with unweighted or non-negative integer weighted edges. 
We present algorithmic upper and lower bounds for five output types:

\begin{itemize}
	\item \textit{Unconstrained APSD enumeration}:
		Algorithms have to output all elements $(u, v, \dist(u, v))$ with $u, v \in V$.
		This precisely mimics the standard APSD problem.
	\item \textit{Row-wise-APSD enumeration}:
		The elements $(u, v, \dist(u, v))$ have to be enumerated sorted by the source vertices $u$.
		This is essentially equivalent to consecutive runs of SSSD enumeration.
	\item \textit{No-self-APSD enumeration}:
		The elements $(u, u, 0)$ have to be omitted in the enumeration.
	\item \textit{Reachable-APSD enumeration}:
		The elements $(u, v, \infty)$ have to be omitted in the enumeration.
		Note that such elements only exist if the graph is not (strongly) connected.
	\item \textit{Sorted-APSD enumeration}:
		The elements $(u, v, \dist(u, v))$ have to be enumerated sorted by the distance $\dist(u, v)$ in increasing order.
\end{itemize}

We analyze both the five individual types as well as all their combinations.
The resulting upper bounds are summarized in \autoref{tab:upper_bds}.
Likewise, \autoref{tab:lower-bds} gives an overview of our lower bounds, some of which are unconditional whilst others are based on popular conjectures or on limited computational models.
All bounds are given in terms of average ($\avgdeg$) and maximum ($\maxdeg$) vertex degree of the input graph.
Combinations not listed explicitly are included in one of the other types:
Adding the restrictions of the latter three output types to row-wise-APSD does not have any impact on the problem's complexity.
Furthermore, sorted-APSD is a strictly stronger requirement than reachable-APSD.

In traditional asymptotic complexity analysis, the time for initializing memory for data structures, that are fully read or written later, is negligible and thus mostly ignored.
However, this no longer applies if for computing the first output a fully prepared data structure is needed, that cannot be initialized within preprocessing and one unit of delay.
For those data structures we make use of lazy initialization (see, e.\,g., the textbook~\cite{MoretShapiroBook} and the detailed description in \autoref{sec:model}).
All our upper bounds specify the exact amount of lazy-initialized memory along with the overall space complexity.
Lazy-initialized memory could be replaced with balanced search trees at an additional cost of $\log(n)$ steps for each read or write operation.
This roughly increases the delay by a factor of $\log(n)$.
All our lower bounds are independent of available memory.

\begin{table}[tb]
	\caption{%
		Upper Bounds for APSD enumeration with different output types, where \emph{row-wise} also transfers to SSSD.
		For each type we present time bounds (order for delay, order for preprocessing) and space bounds (order of lazy-initialized memory [or ``$-$'' if none is needed], order of overall space complexity).
		Bounds in \textbf{bold} are complemented with lower bounds, listed in more detail in \autoref{tab:lower-bds}.
		Results are clickable references to the respective theorem.
	}
	\label{tab:upper_bds}
	\begin{tabularx}{\textwidth}{lYccYcc}
		\toprule
		\multirow{2}{*}{\textbf{Output Type}} &  &                                    \multicolumn{2}{c}{\textbf{Unweighted Graphs}}                                    &  &                                      \multicolumn{2}{c}{\textbf{Non-negative Edge Weights}}                                      \\
		\cmidrule{3-4}\cmidrule{6-7}          &  &                               time                               &                       space                       &  &                                     time                                     &                       space                       \\ \midrule
		unconstrained                         &  &            \hyperref[thm:u-APSD-upper]{$(\avgdeg,1)$}            &       \hyperref[thm:u-APSD-upper]{$(-,n)$}        &  &             \hyperref[thm:w-APSD-upper]{$(\avgdeg + \log(n),1)$}             &       \hyperref[thm:w-APSD-upper]{$(-,n)$}        \\
		row-wise (SSSD)                       &  &     \hyperref[thm:u-SSSD-upper]{$(\boldsymbol{\maxdeg},1)$}      &       \hyperref[thm:u-SSSD-upper]{$(n,n)$}        &  &             \hyperref[thm:w-SSSD-upper]{$(\maxdeg + \log(n),1)$}             &       \hyperref[thm:w-SSSD-upper]{$(n,n)$}        \\
		no-self                               &  &        \hyperref[thm:u-APSD-noself-upper]{$(\avgdeg,1)$}         &    \hyperref[thm:u-APSD-noself-upper]{$(-,n)$}    &  &         \hyperref[thm:w-APSD-noself-upper]{$(\avgdeg + \log(n),n)$}          &    \hyperref[thm:u-APSD-noself-upper]{$(-,n)$}    \\
		reachable                             &  &               \hyperref[cor:r-apsd]{$(\maxdeg,1)$}               &          \hyperref[cor:r-apsd]{$(n,n)$}           &  &                \hyperref[cor:r-apsd]{$(\maxdeg + \log(n),1)$}                &          \hyperref[cor:r-apsd]{$(n,n)$}           \\
		reachable \& no-self                  &  &               \hyperref[cor:r-apsd]{$(\maxdeg,n)$}               &          \hyperref[cor:r-apsd]{$(-,n)$}           &  &                \hyperref[cor:r-apsd]{$(\maxdeg + \log(n),n)$}                &          \hyperref[cor:r-apsd]{$(-,n)$}           \\
		sorted                                &  &           \hyperref[thm:u-s-APSD-upper]{$(\maxdeg,1)$}           &    \hyperref[thm:u-s-APSD-upper]{$(n^2,n^2)$}     &  &            \hyperref[thm:w-s-APSD-upper]{$(\maxdeg + \log(n),1)$}            &    \hyperref[thm:w-s-APSD-upper]{$(n^2,n^2)$}     \\
		sorted \& no-self                     &  & \hyperref[thm:u-s-APSD-noself-upper]{$(\maxdeg,\boldsymbol{n})$} & \hyperref[thm:u-s-APSD-noself-upper]{$(n^2,n^2)$} &  & \hyperref[thm:w-s-APSD-noself-upper]{$(\maxdeg + \log(n),\boldsymbol{n+m})$} & \hyperref[thm:w-s-APSD-noself-upper]{$(n^2,n^2)$} \\ \bottomrule
	\end{tabularx}
\end{table}

\begin{table}[tb]
	\caption{%
		Lower Bounds for APSD enumeration with different output format, where \emph{row-wise} also transfers to SSSD.
		Results are clickable references to the respective theorem.
		(p-c: path-comparison-based, u: unconditional, BMM: unless breakthrough in Boolean matrix multiplication)
	}
	\label{tab:lower-bds}
	\begin{tabularx}{\textwidth}{l@{\hskip 1mm}YY}
		\toprule
		\textbf{Output Type}   & \textbf{Unweighted Graphs}                                                   & \textbf{Non-negative Edge Weights}                                               \\ \midrule
		unconstrained          & -                                                                            & \hyperref[subsec:w-APSD]{$\Omega(\avgdeg)$ delay (p-c)}                          \\
		row-wise (SSSD)        &                                           \multicolumn{2}{c}{\hyperref[thm:SSSD-lower]{$\Omega(\maxdeg)$ delay (u)}}                                            \\
		no-self                & -                                                                            & \hyperref[thm:w-APSD-noself-lower]{$\Omega(\maxdeg)$ preprocessing or delay (u)} \\
		reachable (\& no-self) &                          \multicolumn{2}{c}{\hyperref[thm:r-APSD-lower]{$\omega(\avgdeg)$ delay with preprocessing in $O(n+m)$ (BMM)}}                          \\
		sorted                 &                          \multicolumn{2}{c}{\hyperref[cor:s-APSD-lower]{$\omega(\avgdeg)$ delay with preprocessing in $O(n+m)$ (BMM)}}                          \\
		sorted \& no-self      & \hyperref[thm:u-s-APSD-noself-lower]{$\Omega(n)$ preprocessing or delay (u)} & \hyperref[thm:w-s-APSD-noself-lower]{$\Omega(n+m)$ preprocessing or delay (u)}   \\ \bottomrule
	\end{tabularx}
\end{table}

Interestingly, whether or not edges are directed does not seem to alter the difficulty of these enumeration problems; all our results hold for both undirected and directed graphs.
In this regard, observe also that if we redefined the output of APSD enumeration to be distances for \emph{unordered} pairs of vertices, the delay would only increase by a factor of at most~2.
We discuss the details of this adapted definition in \autoref{sec:undirected-APSD}.

We show a tight delay of $\Theta(\maxdeg)$ for unweighted row-wise-APSD enumeration with an unconditional adversary-based lower-bound, complemented by an algorithm that meets this delay.
Surprisingly, we can beat this bound for unconstrained- and no-self-APSD enumeration, where we achieve a delay in $O(\avgdeg)$.
This means that deviation from the order in which a breadth-first search (BFS) fixes distances cannot be used to improve SSSD enumeration but does have an effect on APSD.
For sorted- and reachable-APSD, an equal improvement would be surprising, as this would provide a quadratic-time algorithm for Boolean matrix multiplication (BMM).

Generally, we solve the weighted problems with additional $O(\log(n))$ cost on the delay compared to their unweighted counterparts, which is due to the cost of minimum extraction on standard priority queues.
We discuss possible speedups through faster integer priority queues along with the respective algorithmic results in \autoref{subsec:w-SSSD} and \autoref{cor:fast-w-APSD-upper}.

In three considered cases it seems favorable to make use of additional preprocessing time.
Two cases (sorted- \& no-self- APSD enumeration) are backed up by a matching unconditional lower bound on either preprocessing or delay.
The third case (weighted no-self-APSD enumeration) leaves a gap between the $\Omega(\Delta)$ lower bound and the $O(n)$ preprocessing time required for our algorithm.
This and further open questions are summarized in \autoref{sec:conclusion}.

We start with the results for row-wise APSD enumeration (i.\,e., SSSD enumeration), in \autoref{sec:SSSD}, as these nicely illustrate the techniques and change of perspective required to approach APSD in the new light of enumeration.
Then we drop the row-wise requirement (i.\,e., we consider actual APSD enumeration), which allows more freedom to distribute the overall computational effort over the enumerated elements in a clever way and therefore leads to better delays.
We exemplify this idea in \autoref{sec:APSD} by using the trivial self-distances to initially buy a head start of computation time.
Then, in \autoref{sec:c-APSD} and \ref{sec:s-APSD}, we extend this general approach to the technically more involved results for the other output types (including APSD enumeration without self-distances).
Before diving into the first results, we formally define some notation and the computational model in the next section.

\section{Algorithmic Model and Basic Definitions}
\label{sec:model}
For all problems in this paper the input is a graph $G = (V,E)$ with $n$ vertices $V = \{ 1, \ldots, n\}$ and $|E| = m$ edges.
The degree of a vertex $v$ is denoted by $\deg(v)$ and we write $\maxdeg$ for the maximum and $\avgdeg$ for the average degree in $G$.
All proofs work on both undirected and directed graphs; by \textit{degree} we always refer to the number of outgoing edges.

Unless otherwise stated, graphs are unweighted.
In this case the shortest $s{-}t$--distance $\dist(s,t)$ is the minimum number of (directed) edge hops to reach $t$ from $s$.
For graphs with edge weights given by a weight function $w\colon E \to \{0, \ldots, n^c\}$ for some constant $c \in \mathbb{N}$, the shortest distance $\dist(s,t)$ is the minimum sum of edge weights on a (directed) path from~$s$ to~$t$.
For a vertex $t$ that is not reachable from vertex $s$ we define $\dist(s,t) = \infty$.
Note that, for both weighted and unweighted graphs, we have $\dist(s, s) = 0$ for every $s \in V$.
We denote such elements $(s, s, \dist(s, s))$ as \emph{self-distances}.

We study algorithms that \emph{enumerate} APSD in the sense that for each $s{-}t$--distance the triple $(s,t,\dist(s,t))$ is emitted exactly once.
We say that such an algorithm has \textit{delay} in $O(f)$, if both the time to compute the first or next output and the time to detect that no further output exists, is in $O(f)$.
On top of that we allow for additional \textit{preprocessing time} the algorithm can spend before the first delay starts.

Algorithms are not required to emit an output immediately after it is computed.
Quite the contrary:
Holding back outputs until the current delay expires is crucial to minimize the worst-case delay of an algorithm.
We analyze this behavior with a variant of the accounting method:
The initial account balance is in the order of the preprocessing time plus one unit of delay.
Every emitted output raises the account balance by at most the aspired delay; we say the output is \textit{charged credit} for this increase; the unit of credit being some implementation specific constant.
This credit is used to pay for all computation steps and must never become negative.
As we expect the algorithm to actively emit the respective next output, the delay (and thus all credit changes) must be computable for the algorithm.

We perform our analyzes in the word RAM model with a word size in $\Omega(\log(n))$.
This allows for all read/write operations and basic arithmetic operations such as comparison and addition of edge weights to be performed in constant time.
Input graphs are encoded as read-only adjacency lists with constant time access to the number of vertices and the degree of each vertex.
This is for example supported by the \textit{standard representation} suggested by Kammer and Sajenko~\cite{kammerLinearTimeInPlaceDFS2019a}.

Algorithms can reserve uninitialized memory within constant time.
To be able to distinguish written memory cells and uninitialized ones in a data structure $A$ one level of indirection is introduced:
The algorithm allocates auxiliary memory $B$ of twice the size of $A$ along with a counter $c$, starting at 0.
$B$ will hold pairs of pointers and actual values, the pairs being written from left to right, whilst $c$ keeps track of the number of pairs in $B$.
Instead of writing and reading from $A$ directly, memory is used as follows:
An initialized cell $A[x]$ is a pointer to a pair in $B$, which in turn holds a pointer to $A[x]$ and the actual value associated with $A[x]$.
Uninitialized memory cells in $A$ either point outside the valid area of $B$ (as specified by $c$) or do not match with the reverse pointer stored in the referenced pair.

\section{Single Source Shortest Distance Enumeration}
\label{sec:SSSD}
In this section, we consider the following enumeration problem:
\begin{problem}[Single Source Shortest Distance (SSSD) Enumeration]~\newline
	\begin{tabularx}{\textwidth}{ll}
		Input:  & Graph $G = (V, E)$, vertex $s \in V$.\\
		Output: & For each $t \in V$ the shortest distance solution $(t, \dist(s, t))$.
	\end{tabularx}
\lipicsEnd
\end{problem}

Note that running an SSSD enumeration algorithm for each vertex solves row-wise-APSD enumeration.
Conversely, to solve SSSD enumeration on the source vertex $s$ with an algorithm for row-wise-APSD enumeration, one swaps each access to vertex~1 with an access to vertex~$s$ and vice versa and stops the enumeration once the algorithm produced the complete first row of the distance matrix.
These reductions transfer all lower and upper bounds between SSSD and row-wise-APSD enumeration.
Thus, instead of analyzing row-wise-APSD, we inspect the equivalent SSSD enumeration problem. 

\subsection{Unweighted Graphs}
\label{subsec:u-SSSD}

We first consider the case of unweighted graphs, for which a single breadth-first search (BFS) started in vertex~$s$ solves SSSD in optimal total time $O(n + m)$ and with $\Theta(n)$ space. It turns out that a BFS can also be carried out in such a way that it produces the distance triples $(s, t, \dist(s, t))$ with a delay bounded by the maximum degree of the graph (\autoref{thm:u-SSSD-upper}).
Moreover, we can also show that this upper bound on the delay is essentially optimal (\autoref{thm:SSSD-lower}).

\begin{theorem}
\label{thm:u-SSSD-upper}
	Unweighted SSSD enumeration can be solved with delay in $O(\maxdeg)$ with $\Theta(n)$ lazy-initialized memory and space complexity in $\Theta(n)$.
\end{theorem}
\begin{proof}
	Our modified BFS version works as follows:
	It maintains a solution queue~$Q$ to collect computed distances for later output.
	Whenever the search visits a vertex~$t$, it processes all its neighbors and, when indicated, updates their distance and enqueues them to be visited themselves later.
	Afterwards, the algorithm appends $(t, \dist(s, t))$ to~$Q$.

	After the BFS is finished, the algorithm iterates over the vertices once and appends for all unvisited vertices $v$ the solution~$(v, \infty)$ to~$Q$.
	(Note that this final step takes $\Theta(n)$ time while perhaps not producing further output; this happens if there are no unreachable vertices, which is why we actually need the solution queue.)

	During the breadth-first search, the algorithm tracks the current maximum node degree $\maxdeg_k$ seen so far.
	For some implementation specific constant $c$ and the respective current $\maxdeg_k$, it emits the next solution from~$Q$ after $c \cdot \maxdeg_k \in O(\maxdeg)$ steps.

	Whenever the algorithm is required to emit an output, $Q$ must not be empty.
	We apply the accounting method to prove this, effectively tracking how many solutions are in~$Q$ and how much computation time can be spent before a new output must be emitted.
	For each credit unit, the algorithm can perform a constant number of steps.
	As long as the credit stays positive, $Q$ is not empty.

	By induction over the number~$k$ of vertices visited by the BFS we show that (1) the account balance is always non-negative, (2) it is at least $\maxdeg + k$ after the $k$-th vertex has been visited and (3) we achieve this by charging each output with at most $\maxdeg + 1$ credit.
	For $k=0$ this is true, as the initial balance is $\maxdeg$ as specified in \autoref{sec:model}.
	Let the claims be true for some $0 \leq k \leq n-1$ and let $t$ be the $(k+1)$st vertex visited by the BFS.
	It takes $O(\deg(t))$ steps to process the neighbors of~$t$ and to put the solution $(t,\dist(s,t))$  into~$Q$, which is payed for with $\deg(t) \leq \maxdeg$ units.
	Recall that the algorithm tracks the maximum node degree $\maxdeg_{k+1}$ seen so far.
	The newly produced solution is charged $\maxdeg_{k+1} + 1$ units.
	As $\maxdeg_{k+1} \geq \deg(t)$ holds, the new credit balance is at least $\maxdeg + k + 1$ as claimed.

	When the last vertex was visited by the BFS, the account balance is at least $\maxdeg + r$, where $r$ is the number of vertices reachable from $s$.
	In the final iteration over all vertices, the algorithm uses this balance to skip the $r$ already seen vertices.
	The remaining $n-r$ solutions with distance $\infty$ directly pay for the rest of the iteration.

	As each output is charged at most $\maxdeg + 1$ credit and the account balance is never negative, the algorithm enumerates all solutions to SSSD with delay in $O(\maxdeg)$.

	The BFS uses $\Theta(n)$ lazy-initialized memory to keep track of the distance, predecessor and visited state of the vertices.
	Both the internal queue of the BFS and the solution queue $Q$ hold at most $n$ entries.
	Thus, the overall space complexity of the algorithm is in $\Theta(n)$.
\end{proof}

Before presenting the lower bound, we note that (on a connected graph) a BFS spends on average $O(\avgdeg)$ time per discovered vertex.
Consequently, the question arises whether $O(\avgdeg)$ can also be achieved as the delay in the enumeration setting.
The following lower bound, however, shows that this is not possible and that the delay of $O(\maxdeg)$ achieved by the algorithm of \autoref{thm:u-SSSD-upper} is optimal.
The intuitive idea of the lower bound is that any algorithm for SSSD enumeration can be forced to fully process first a potentially densely connected neighborhood of $s$ with many edges but few solutions, before progressing to areas of the graph where less effort is needed to produce many solutions.

\begin{theorem}
\label{thm:SSSD-lower}
	Unweighted SSSD enumeration cannot be solved with delay in $o(\maxdeg)$, even if the graph is connected and has constant average degree.
\end{theorem}
\begin{proof}
	Assume some algorithm $A$ was able to enumerate the distances with delay in $o(\maxdeg)$ and consider the following adversarial setup that is equivalent to receiving the input graph as adjacency lists:
	$A$ is allowed to ask the adversary for (a) the next neighbor of any vertex (and thereby iterate through its  adjacency list) and (b) the degree of any vertex.

	The adversary will, for arbitrary $k$, construct a graph consisting of a $k$-clique $C = \{ v_1, \ldots, v_k \}$ with the start vertex $s$ in it.
	One edge of this clique is replaced by a path of length $k^2 + 1$ with $k^2$ inner vertices $P = \{u_1, \ldots, u_{k^2}\}$.
	\Autoref{fig:SSSD-lower} depicts the graph structure.
	Note, that the graph has $|V| = k + k^2$ vertices and $\frac{k (k - 1)}{2} + k^2$ edges and thus constant average degree.

	When asked for a vertex degree, the adversary consistently reports $\deg(v) = k - 1$ for all $v \in C$ and $\deg(u) = 2$ for all $u \in P$.
	For each neighborhood query of a vertex $v \in C$, the adversary only returns neighbors within the clique, until there are only~$2$ clique vertices left, for which the algorithm has not queried the complete neighborhood.
	Let $c_1, c_2$ be these two clique vertices.
	Similarly, for each neighborhood query of a vertex $u \in P$, the adversary only returns neighbors within the path, until there are only $2$ path vertices with unseen neighborhood.
	Let $p_1, p_2$ be these two path vertices.
	Note that the algorithm has to make at least $(k - 1)^2$ queries to the adversary until it sees the first adjacent edge to one of $p_1, p_2$;
	either by asking at least $k-1$ complete neighborhoods of vertices in $C$, or at least $k^2-3$ neighborhoods of vertices in $P$.

	Now assume the algorithm emits for any  $t \in P$ a shortest distance solution $(t,\dist(s,t))$ before having made at least $(k-1)^2$ queries.
	As no adjacent edge of neither $p_1$ nor $p_2$ was fixed yet, the adversary is free to position the two vertices as shown in \autoref{fig:SSSD-lower} or instead move $p_1$ in-between $c_2$ and $p_2$ or move $p_2$ in-between $c_1$ and $p_1$.
	For at least one of those three configurations, $(t,\dist(s,t))$ is not a correct solution:
	If $\dist(s,t)$ is the correct solution for the configuration shown in \autoref{fig:SSSD-lower} and a shortest path from $s$ to $t$ w.l.o.g. enters $P$ through the vertices $c_1, p_1$, then moving $p_1$ in-between $c_2$ and $p_2$ decreases the shortest distance by~1 and thereby renders the solution invalid.
	In the second case, $\dist(s,t)$ is incorrect for the arrangement in \autoref{fig:SSSD-lower} and the adversary can simply choose this configuration.

	Thus, the algorithm needs $\Omega(k^2)$ queries to the adversary before being able to correctly determine a distance to any vertex from $P$.
	As there are only $k$ vertices outside $P$ (and hence only $k$ outputs), this results in an enumeration delay of at least $\Omega(k)$.
	The constructed graph has maximum degree $\maxdeg = k-1$, thus the delay is in $\Omega(\maxdeg)$.
	\begin{figure}[t]
		\centering
		\includegraphics{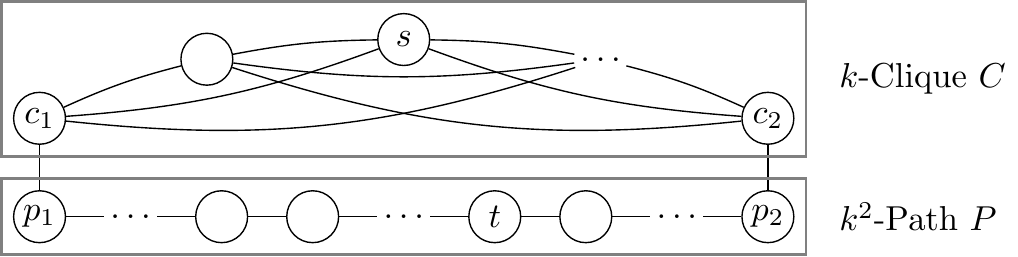}
		\caption{The adversary for \autoref{thm:SSSD-lower} constructs a clique with one edge replaced by a long path.}
		\label{fig:SSSD-lower}
	\end{figure}
\end{proof}

\subsection{Non-negative Edge Weights}
\label{subsec:w-SSSD}

Next, we consider graphs with non-negative integer edge weights.
In this setting, Dijkstra's algorithm~\cite{dijkstraNoteTwoProblems1959} using a Fibonacci Heap as priority queue~\cite{fredmanFibonacciHeapsTheir1987} solves the SSSD problem in $O(m + n \cdot \log(n))$ with $\Theta(n)$ space.
Similarly as we extended BFS to an enumeration algorithm in \autoref{thm:u-SSSD-upper}, we can extend Dijkstra's algorithm to an algorithm for SSSD enumeration for weighted graphs.

\begin{theorem}
	\label{thm:w-SSSD-upper}
	SSSD enumeration on graphs with edge weights in $\{0, \ldots, n^c\}$, for some constant $c$, can be solved with delay in $O(\maxdeg + \log(n))$  with $\Theta(n)$ lazy-initialized memory and space complexity in $\Theta(n)$.
\end{theorem}
\begin{proof}
	Our enumeration algorithm replaces BFS with Dijkstra's algorithm in \autoref{thm:u-SSSD-upper}.
	By using a priority queue such as Strict Fibonacci Heaps \cite{brodalStrictFibonacciHeaps2012} with $O(\log(n))$ worst case time for \textsc{extractMin} and $O(1)$ worst case time for \textsc{makeHeap}, \textsc{insert} and \textsc{decreaseKey}, we get a delay of $O(\log(n))$ for determining the next vertex $t$ and its shortest distance $\dist(s,t)$, and additionally $\maxdeg \cdot O(1)$ delay for processing the neighborhood of $t$.
\end{proof}

We next argue that this algorithm is an optimal enumeration variant of Dijkstra's algorithm with a Fibonacci Heap.
First, we note that the lower bound in \autoref{thm:SSSD-lower} also applies for weighted graphs.
Consequently, if $\maxdeg \in \Omega(\log(n))$, then the delay $O(\maxdeg + \log(n)) = O(\maxdeg)$ is optimal.
If, on the other hand, the maximum degree is smaller, then our enumeration version of Dijkstra's algorithm has delay $O(\log(n))$, but also the number of edges is in $O(n \log(n))$, which means that in the non-enumeration setting, Dijkstra's algorithm also needs on average $O(\log(n))$ time per computed distance triple.

There are several algorithms for SSSD with non-negative integer weights that outperform Dijkstra's algorithm in the theoretical analysis, such as Thorup's algorithm with a total runtime of $O(m + n(\log\log(n))^{1+\varepsilon})$ in our setting.
However, it is not clear whether such improvements indicate the existence of enumeration algorithms with smaller delay for SSSD, as they are usually based on priority queues with much higher initialization overhead.
In contrast, such algorithms \emph{can} be used to enumerate APSD faster, as explained in \autoref{subsec:w-APSD}.

\subsection{Constrained SSSD Enumeration}

Any standard SSSD algorithm would produce as output an array of length $n$ that includes the shortest distance from $s$ to all vertices of the input graph, including $s$ itself and all vertices unreachable from~$s$.
For graphs with only non-negative edge weights the former distance is always $0$ (independent of the input graph) and can therefore considered negligible.
The distances to unreachable vertices are all $\infty$ and, while the actual unreachable vertices obviously do depend on the input graph, they may be completely omitted from the output, since they are nevertheless implicitly given by omission (i.\,e., if we know that distances to unreachable vertices are not part of the output, they are all still uniquely represented by the output).
In addition to omitting redundant information in the output, we could also require the distance triples to be sorted with respect to the distance.
The question arises how these modifications influence the upper and lower delay bounds shown so far.

In fact, the next corollary states that the above mentioned modifications do not change our upper delay bounds (and this also holds for all possible combinations of them):
Regarding the omission of self-distances or of distance triples of unreachable vertices, our algorithms can just skip the single self-distance by doubling the delay and simply stop enumerating when all reachable vertices have been fully processed.
Moreover, the requirement of producing the distance-triples sorted by the distance is already fulfilled by both BFS and Dijkstra's algorithm.
With the terminology defined in \autoref{sec:contribution} this can be summarized as follows.

\begin{corollary}
	\label{cor:c-sssd}
	No-self-, reachable- and sorted-SSSD enumeration as well as all combinations of the three output types can be solved with delay in $O(\maxdeg)$ for unweighted graphs and $O(\maxdeg + \log(n))$ for graphs with edge weights in $\{0, \ldots, n^c\}$, for some constant $c$.
	Both results use $\Theta(n)$ lazy-initialized memory and have an overall space complexity of $\Theta(n)$.
\end{corollary}

With respect to the lower bounds, we can first observe that the lower bound from \autoref{thm:SSSD-lower} still applies if we omit self-distances, or distances of unreachable vertices, or both.
For sorted enumeration, we inherit an $\Omega(\log(n))$ lower bound on the delay in the addition-comparison model from the $\Omega(n \log(n))$ lower bound for comparison-based sorting, as enumerating the distances of a star center to its spikes sorted by distance is equivalent to sorting the edge weights.

\begin{corollary}
	\label{cor:w-sorted-sssd-lower}
	In the addition-comparison model, no algorithm can solve sorted-SSSD enumeration on graphs with non-negative edge weights with preprocessing in $o(n\log(n))$ and delay in $o(\log(n))$.
\end{corollary}

In summary, omitting the redundant self-distances or distances to unreachable vertices, or requiring sorted enumeration has no influence on SSSD enumeration, but, as we shall see in Sections~\ref{sec:c-APSD} and \ref{sec:s-APSD}, the situation is quite different for APSD enumeration.

\section{Unconstrained All Pairs Shortest Distance Enumeration}
\label{sec:APSD}

For APSD enumeration, we require that all entries of the distance matrix are listed:

\begin{problem}[All Pairs Shortest Distance (APSD) Enumeration]~\newline
	\begin{tabularx}{\textwidth}{ll}
		Input:  & Graph $G = (V, E)$.\\
		Output: & For each $s, t \in V$ the shortest distance solution $(s, t, \dist(s, t))$.
	\end{tabularx}
\lipicsEnd
\end{problem}
Again, we start with unweighted graphs and then extend to non-negative edge weights.

\subsection{Unweighted Graphs}
\label{subsec:u-APSD}

Since solving the APSD problem reduces to solving SSSD for each vertex $s \in V$ as start vertex, it is obvious that Theorem~\ref{thm:u-SSSD-upper} extends to APSD, i.\,e., unconstrained APSD enumeration is possible with delay $O(\maxdeg)$.
However, such an algorithm just performs $n$ breadth-first searches in total time $O(n^2 + nm)$ and therefore spends on average $O(\avgdeg)$ time per produced distance triple.
Hence, just like discussed for SSSD after the proof of Theorem~\ref{thm:u-SSSD-upper}, the question arises whether $O(\avgdeg)$ can also be achieved as delay for APSD enumeration.
While we answered this question in the negative for SSSD by presenting the lower bound of Theorem~\ref{thm:SSSD-lower}, we can actually achieve delay $O(\avgdeg)$ for APSD enumeration (this means that for APSD enumeration, we can do better than just running the SSSD enumeration algorithm $n$ times).

\begin{restatable}{theorem}{thmUApsdUpper}
	\label{thm:u-APSD-upper}
	Unweighted APSD enumeration can be solved with delay in $O(\avgdeg)$ and with space complexity in $\Theta(n)$.
\end{restatable}
\begin{proof}
	Our algorithm maintains a queue $Q$ of known solutions from which it one by one emits a shortest distance whenever the current delay comes to an end.
	As first step, the algorithm iterates over all vertices, appending to $Q$ all solutions with distance 0.
	All longer distances are found and appended to $Q$ through breadth-first searches on all vertices.
	Because the algorithm initially does not know the average degree $\avgdeg$, it emits the first $\frac{n}{2}$ solutions in~$Q$ with constant delay and in parallel computes $\avgdeg$ by summing up node degrees and dividing by $n$.
	From then on, all further entries in $Q$ will be emitted with $\Theta(\avgdeg)$ delay.

	We again apply the accounting method to prove that $Q$ is never empty when a solution has to be emitted.
	Clearly, the first $\frac{n}{2}$ solutions (distance 0) can be produced and emitted with constant delay.
	We assign a payment of 3 to each of those solutions, whilst the actual cost is 1, leaving a credit of $n$ which is used to compute $\avgdeg$ (recall that we assume constant access to the degree of a vertex, thus $\avgdeg$ can be computed in $O(n)$ time).
	The remaining $\frac{n}{2}$ trivial solutions are also produced with an actual cost of 1 but are in addition charged~$4 \cdot (\avgdeg + 1)$, resulting in $2n(\avgdeg + 1)$ credit.
	All further solutions produced by the breadth-first searches are assigned a cost of $\avgdeg$.

	We now show that the algorithm can run BFS on all vertices without running out of credit.
	For this we claim that, after the algorithm has performed a breadth-first search for each start vertex in the set $U \subseteq V$, the credit balance is at least $B(U) = |V \setminus U| \cdot (\avgdeg + 1) + m + n$.

	Recall that $m = n \cdot \avgdeg$, thus the claim is true for $U = \emptyset$.
	Assuming the claim holds for any $U \subset V$ we inspect a breadth-first search on any vertex $s \in V \setminus U$.
	As the search costs $m+n$ and we have $B(U) > m + n$, the BFS can be performed without the solution queue running empty in-between.
	The search produces $n - 1$ new solutions which are each charged with a cost of $\avgdeg + 1$.
	(Additionally, it skips one solution with distance $0$ that has already been put into $Q$ during initialization.)
	Thus we get $B(U \cup \{s\}) = B(U) - (m + n) + (n - 1) \cdot (\avgdeg + 1)$ as claimed.

	As each solution is assigned a cost in $O(\avgdeg)$, the algorithm solves APSD enumeration with delay in $O(\avgdeg)$.

	By alternating between two sets of data structures for the breadth-first searches, only $\Theta(n)$ memory is needed for all BFS runs.
	During each BFS the algorithm can prepare the memory needed for the subsequent search; same holds for the first BFS whose memory can be prepared during the initial pass over the vertices.
	Note that the solution queue $Q$ never has to hold more than $O(n)$ solutions, as for all $U \subseteq V$ we have $B(U) / \avgdeg \in O(n)$.
	Thus, the algorithm has a space complexity of $\Theta(n)$.
\end{proof}

In the proof of \autoref{thm:u-APSD-upper}, we used the distance triples $(v,v,0)$ in order to get a head start that allowed us to carry out the BFSs in such a way that the delay of $O(\avgdeg)$ is maintained.
But just like the distance triples $(v,v,0)$ can be used for this purpose, we could also use the trivial distance triples $(u,v,1)$ that are directly given by the edges $(u, v) \in E$.
Assuming that we aim for maintaining a delay bound of $D$, this would give us a head start of $\Theta(m \cdot D)$ computation time.
Now for dense graphs, $m$ may be large enough such that we can fully solve the APSD problem within time $O(m \cdot D)$ by means of fast matrix multiplication (FMM), for example with the algorithm by Seidel \cite{seidelAllPairsShortestPathProblemUnweighted1995}, and afterwards emit the solutions with distance $> 1$ with constant delay.
This observation is formally stated in the next remark.
\begin{remark}
	\label{rm:fmm-APSD}
	Unweighted APSD enumeration can be solved with delay in $O(\min(\avgdeg, \frac{n^\omega\log(n)}{m}))$, where $\omega < 2.37286$  is the matrix multiplication constant~\cite{almanRefinedLaserMethod2021}.
\end{remark}

\subsection{Non-negative Edge Weights}
\label{subsec:w-APSD}

We can tackle non-negative edge weights by again replacing BFS with Dijkstra's algorithm.
In comparison to \autoref{thm:u-APSD-upper}, we only need  to adjust the credit payments assigned to the solutions to match the aspired delay.

\begin{restatable}{theorem}{thmWApsdUpper}
	\label{thm:w-APSD-upper}
	APSD enumeration on graphs with edge weights in $\{0, \ldots, n^c\}$, for some constant $c$, can be solved with delay in $O(\avgdeg + \log(n))$ and with space complexity in $\Theta(n)$.
\end{restatable}
\begin{proof}
	If BFS is replaced by Dijkstra's algorithm in \autoref{thm:u-APSD-upper}, the analysis needs to be changed as follows:

	The second half of the trivial solutions is charged, in addition to the constant cost for producing it, with $4(\avgdeg + \log(n))$, resulting in $2n(\avgdeg + \log(n))$ units on the account.
	We now have $B(U) = |V \setminus U| \cdot (\avgdeg + \log(n)) + m + n\log(n)$.
	For $U \subset V$ this credit is enough to fully run Dijkstra's algorithm with cost $m + n \log(n)$, producing $n-1$ new solutions which are each charged with $\avgdeg + \log(n)$ units.
	Thus, $B(U)$ is non-negative for all $U \subseteq V$ and the algorithm has delay in $O(\avgdeg + \log(n))$.
	The space analysis from \autoref{thm:u-APSD-upper} applies without changes, resulting in $\Theta(n)$ space complexity.
\end{proof}

In case $\avgdeg \in \Omega(\log(n))$ holds, this delay is optimal in the path-comparison model, as it matches the total time lower bound of $\Omega(mn)$ shown by Karger~et~al.~\cite{kargerFindingHiddenPath1993}.
For other models one can use improved SSSD algorithms instead of Dijkstra's algorithm to achieve a better delay.
If one allows the full power of word RAM \emph{including} multiplication this goes as far as using Thorup's $O(m)$ shortest paths algorithm for non-negative integer weights \cite{thorupUndirectedSinglesourceShortest1999}, which gives a delay of $O(\avgdeg)$ for APSD enumeration on undirected graphs.

\begin{corollary}
	\label{cor:fast-w-APSD-upper}
	APSD enumeration on graphs with edge weights in $\{0, \ldots, n^c\}$, for some constant $c$, can be solved with delay in $O(\frac{1}{n} \cdot T)$, where $T$ is the runtime of an SSSD algorithm and is computable in $O(n)$ time.
\end{corollary}

\section{Constrained All Pairs Shortest Distance Enumeration}
\label{sec:c-APSD}

Standard APSD algorithms produce a complete distance matrix and our previous enumeration algorithms mimic this by enumerating the values of all matrix cells in the form of all distance triples $(u, v, d(u, v))$.
In the following, we consider the output restrictions discussed at the beginning of Section~\ref{sec:contribution}, i.\,e., no-self-APSD enumeration (where the self-distances $(u, u, 0)$ are omitted) and reachable-APSD enumeration (where distance triples $(u, v, \infty)$ are omitted).

\subsection{No-Self-APSD Enumeration}

In contrast to the SSSD enumeration with delay in $\Omega(\maxdeg)$, the APSD enumeration algorithms of Section~\ref{sec:APSD} buy some head start time by emitting all entries on the main diagonal of the matrix; which are the solutions that indicate the distance of a vertex to itself.
For graphs with negative edge weights these entries are important, as they indicate whether the graph contains a cycle of negative total weight.
With none or only non-negative weights however, they do not depend on the input graph and are thus somewhat uninteresting for the receiver of the enumerated solutions.
So, what happens, if only distances between pairwise different vertices are to be emitted -- in other words, what delay can we achieve for no-self-APSD enumeration?

In the unweighted case, each edge $(u,v) \in E$ of the input graph corresponds to a shortest distance solution $(u,v,1)$; we have used this before in \autoref{rm:fmm-APSD} for solving unweighted APSD enumeration via FMM.
As it takes $\Omega(n)$ steps to find the first edge in the worst case, this is not enough to ensure a delay in $O(\avgdeg)$.
However, vertices with (out-)degree 0 admit for even more easy-to-compute solutions:
For each $s \in V$ with $\deg(s) = 0$ and each $t \in V \setminus \{s\}$ there is no shortest path from $s$ to $t$ and therefore $(s,t,\infty)$ is a solution.
By modifying the algorithm of \autoref{thm:u-APSD-upper} to use these solutions instead of self-distances, we achieve the same delay in the constrained case.

\begin{restatable}{theorem}{thmUApsdNoselfUpper}
	\label{thm:u-APSD-noself-upper}
	Unweighted no-self-APSD enumeration can be solved with delay in $O(\avgdeg)$ and with space complexity in $\Theta(n)$.
\end{restatable}
\begin{proof}
	The algorithm maintains a solution queue $Q$ and a queue $S$ of start vertices for later breadth-first searches.
	In the first step, the algorithm iterates over all vertices and edges, computes $\avgdeg$ and distinguishes two cases for the current vertex $s$:
	\begin{enumerate}
		\item If $\deg(s) = 0$, append for each $t \in V \setminus \{s\}$ the solution $(s,t,\infty)$ to $Q$.
		\item Otherwise, append for each $(s,t) \in E$ the solution $(s,t,1)$ to $Q$ and append $s$ to $S$.
	\end{enumerate}
	Afterwards, the algorithm finds all remaining distances by running BFS on each vertex in $S$ and from each such search emitting all solutions with distance~$>1$.
	The first $\frac{n}{2}$ solutions are emitted with constant delay, all remaining ones with delay in~$\Theta(\avgdeg)$.

	During the first pass over all vertices and edges, the algorithm produces at least $\max(m, n)$ solutions, each within a constant number of steps.
	We again assign a payment of 3 to the first $\frac{n}{2}$ solutions and pay for the computation of $\avgdeg$ with the surplus $n$ units of credit.
	The remaining at least $\max(m, n) - \frac{n}{2} \geq \frac{\max(m, n)}{2}$ initial solutions are assigned an additional payment of $4 \cdot (\avgdeg + 1) + 2 \cdot (\avgdeg + 1) + 4$ units.
	We claim that, after the algorithm has run a BFS for each start vertex in the set $U \subseteq S$, the credit balance is at least
	\[B(U) = \sum_{v \in V \setminus U} \deg(v) \cdot (\avgdeg + 1) + |V \setminus U| \cdot (\avgdeg + 1) + m + n.\]
	For $U = \emptyset$ this is true:
	Recall that, by the handshaking lemma, $2m = \sum_{v \in V} \deg(v)$ holds.
	Assuming the claim holds for any $U \subset V$ we inspect a breadth-first search with cost $m + n$ on any vertex $s \in S \setminus U$.
	As $B(U) > m + n$, the BFS can be performed without the solution queue running empty during the search.
	The BFS computes the distance to all vertices in the graph, producing $n - 1 - \deg(s)$ new solutions which are each charged witch a cost of $\avgdeg + 1$.
	(Additionally, the BFS computes the self-distance and $\deg(s)$ solutions with distance 1 that have already been put into $Q$ during initialization.)
	Thus we get $B(U \cup \{s\}) = B(U) - (m + n) + (n - 1 - \deg(s)) \cdot (\avgdeg + 1)$ as claimed.

	Note that this algorithm produces up to $\Theta(n^2)$ solutions during the first iteration over all vertices and edges which would require $\Theta(n^2)$ memory for the solution queue.
	However, these solutions can be produced on demand whenever the queue runs empty which only requires a single pointer to track the iteration progress.
	Thus, as before in \autoref{thm:u-APSD-upper}, the algorithm has an overall space complexity in $O(n)$.
\end{proof}

For graphs with non-negative weights we can apply a similar idea:
Not all edges form a shortest path, but at least for each vertex the outgoing edge with smallest weight does.
It takes $\Omega(\maxdeg)$ time in the worst case to find such an edge, which is also the minimum time any algorithm must spend before it can possibly emit the first solution.

\begin{restatable}{theorem}{thmWApsdNoselfLower}
	\label{thm:w-APSD-noself-lower}
	No algorithm can solve no-self-APSD enumeration on a weighted graph with both delay and preprocessing in $o(\maxdeg)$, even if the graph is connected and has constant average degree.
\end{restatable}
\begin{proof}
	We extend the adversarial setup from \autoref{thm:SSSD-lower} with the graph structure shown in \autoref{fig:SSSD-lower} with the following additions:
	\begin{itemize}
		\item The first $k$ vertices the algorithm receives any information on (through neighborhood or degree queries) are set to be from the clique $C$.
		\item With each query for a neighbor $v$ of some vertex $u$, the algorithm also receives the edge weight $w(u,v)$.
		\item The first $k-2$ edge weights reported by the adversary are set to $3$.
	\end{itemize}
	Assume, an algorithm reports $(s, t, d)$ as solution after at most $k-2$ queries to the adversary.
	Clearly, the algorithm cannot emit a shortest distance solution containing some unseen vertex, thus $s, t \in C$ must hold.
	If $d \geq 3$, then the adversary sets all remaining edge weights to 1.
	As there must be a common neighbor $v$ of both $s$ and $t$ of which the algorithm has not queried a single edge, the correct shortest distance would have been $(s,t,2)$ with the shortest path $(s,v,t)$.
	For $d \leq 2$ all remaining edge weights are also set to 3 and $\dist(s,t) \geq 3$ holds, again making the output wrong.
	Therefore every algorithm needs to ask at least $k$ queries before it can emit the first shortest distance.
	As $\maxdeg = k-1$, an algorithm needs to have delay or preprocessing time in $\Omega(\maxdeg)$.
\end{proof}

Whilst not fully closing the gap to this $\Omega(\maxdeg)$ lower bound on delay plus preprocessing, we can give an algorithm that uses $O(n)$ preprocessing time and then achieves the same delay as previous algorithms without output constraints:
During preprocessing, the algorithm computes $\avgdeg$ and, using BucketSort, sorts the vertices by increasing degree.
In this order, the algorithm can afterwards find for each vertex the adjacent edge with minimum weight and put the corresponding solution into the solution queue.
This way, the algorithm does not run into vertices with high degree early on, but can build some head start with prepared solutions first.
\begin{restatable}{theorem}{thmWApsdNoselfUpper}
	\label{thm:w-APSD-noself-upper}
	No-self-APSD enumeration on graphs with edge weights in $\{0, \ldots, n^c\}$, for some constant $c$, can be solved with delay in $O(\avgdeg + \log(n))$, preprocessing in $O(n)$ and with space complexity in $\Theta(n)$.
\end{restatable}
\begin{proof}
	During preprocessing, the algorithm sorts the vertices by increasing node degree and computes $\avgdeg$.
	Again, two queues are maintained, $Q$ for solutions and $S$ for start vertices for Dijkstra's algorithm.
	In the enumeration phase the algorithm iterates over the sorted vertices and distinguishes two cases for the current vertex $s$:
	\begin{enumerate}
		\item If $\deg(s) = 0$, append for each $t \in V \setminus \{s\}$ the solution $(s,t,\infty)$ to $Q$.
		\item Else, find $t = {\arg\min}_{v \in V}\{\, w(s,v) \mid (s,v) \in E \,\}$, append $(s,t,w(s,t))$ to $Q$ and $s$ to $S$.
	\end{enumerate}
	The remaining distances are found by running Dijkstra's algorithm on each vertex in $S$.
	In each such run, the respective edge from case~2 can be found again and the corresponding distance solution can be skipped.
	All solutions are emitted with delay in $\Theta(\avgdeg)$.

	By using Bucket Sort, the preprocessing can be done in $O(n)$ time with $O(n)$ memory.
	As the algorithm iterates over the vertices in order of their degree, the initial solutions can be produced with delay in~$O(\avgdeg)$:
	In the first case, the algorithm produces $n-1$ solutions with constant delay, in the second case it takes $O(\deg(s))$ time to compute one solution for vertex $s$.
	Let $v_1, v_2, \ldots, v_n$ be the vertices sorted by increasing (out-)degree.
	Then, for each $1 \leq j \leq n$, the inequality $\sum_{i=1}^{j} \deg(v_i) \leq j \cdot \avgdeg$ holds.
	Thus, by charging each solution with $\avgdeg$ units of credit, the computation effort for the expensive solutions can always be paid for; incorporating the $\infty$-solutions only increases the head start.

	We charge each of the at least $n$ initial solutions additional $3(\avgdeg + \log(n))$ units.
	In comparison to \autoref{thm:w-APSD-upper} we then have $B(U) = 2|V \setminus U| \cdot (\avgdeg + \log(n)) + m + n\log(n)$ and $n-2$ new solutions per run of Dijkstra's algorithm, each charged with $\avgdeg + \log(n)$ units.
	Again, $B(U)$ is non-negative for all $U \subseteq V$ and the algorithm enumerates the solutions with delay in $O(\avgdeg + \log(n))$.

	By using the same trick as in \autoref{thm:u-APSD-noself-upper}, the initial solutions can be produced on demand, without memory-overhead.
	The space complexity of the remaining algorithm is $\Theta(n)$ as before in \autoref{thm:w-APSD-upper}.
\end{proof}

\subsection{Reachable-APSD Enumeration}

Both BFS and Dijkstra's algorithm primarily deal with the vertices reachable from the start vertex.
Thus, for the SSSD algorithms based on BFS and Dijkstra's algorithm, the distance $\infty$ to unreachable and thus unvisited vertices is produced merely as a byproduct.
However, restricting APSD enumeration to distances $< \infty$ (i.\,e., \emph{reachable-APSD}, as defined in \autoref{sec:contribution}) has a big impact on the delay.
In fact, solving this restriction with delay in $O(\avgdeg)$ would imply a surprising breakthrough in Boolean Matrix Multiplication (BMM), as we can reduce BMM to reachable-APSD enumeration:
We augment the standard graph representation of the $d \times d$ BMM instance with additional $O(d^2)$ vertices that reduce the average degree to a constant.
Enumerating the $O(d^2)$ reachable distances with delay in $O(\avgdeg)$ then solves BMM in $O(d^2)$ total time (since the average degree is constant).
\begin{restatable}{theorem}{thmRApsdLower}
	\label{thm:r-APSD-lower}
	No algorithm can solve reachable-APSD enumeration with delay in $O(\avgdeg)$ and preprocessing in $O(n + m)$ unless $d \times d$ BMM can be solved in $O(d^2)$ total time.
\end{restatable}
\begin{proof}
	We reduce Boolean Matrix Multiplication to reachable-APSD enumeration.
	Given two $d \times d$ Boolean matrices $A$ and $B$ we construct the standard graph representation of the BMM instance:
	The graph consists of three vertex partitions $I, J, K$, each with vertices~$\{1, \ldots, d\}$.
	Connect two vertices $i \in I, j \in J$ iff $A[i,j] = 1$ and two vertices $j \in J, k \in K$ iff $B[j,k] = 1$.
	Thus we have, for all $i \in I, k \in K$, the equivalence $(A \cdot B)[i,k] = 1 \Leftrightarrow \dist(i,k) = 2$.

	We augment this graph with $2d^2$ isolated vertices.
	Note that the resulting graph has $n = 2d^2 + 3d$ vertices, $m \leq 2d^2$~edges and therefore constant average degree.
	Building the graph can be done in $O(d^2)$ time.
	The restricted APSD enumeration produces $O(d^2)$ solutions from which the resulting matrix $A \cdot B$ can be constructed in $O(d^2)$ time.
	Thus, an enumeration algorithm for reachable-APSD with delay in $O(\avgdeg)$ and preprocessing in $O(n + m)$ yields an $O(d^2)$ BMM algorithm.
\end{proof}

Recall that by \autoref{cor:c-sssd}, reachable-(no-self)-SSSD can be solved with delay in $O(\maxdeg)$ on unweighted graphs and with delay in $O(\maxdeg + \log(n))$ on graphs with non-negative edge weights.
Repeatedly applying these SSSD enumeration algorithms solves the corresponding APSD-enumeration with the same delay bounds.
With the no-self restriction it may take up to $\Theta(n)$ steps to find the first edge.
\begin{corollary}
	\label{cor:r-apsd}
	Reachable-APSD enumeration can be solved with delay in $O(\maxdeg)$ for unweighted graphs and $O(\maxdeg + \log(n))$ for graphs with edge weights in $\{0, \ldots, n^c\}$, for some constant $c$.
	Both results use $\Theta(n)$ lazy-initialized memory and have an overall space complexity of $\Theta(n)$.
	Without self-distances, $O(n)$ additional preprocessing is needed, which in turn eliminates the need for lazy-initialized memory.
\end{corollary}
Both algorithms have cubic total time for graphs with high maximum degree.
The fastest so-called \textit{combinatorial} algorithm to date multiplies $n \times n$ matrices in $O(n^3 / \log(n)^3)$ time (ignoring $\log\log$ factors)~\cite{yuImprovedCombinatorialAlgorithm2018}.
As BMM can be solved with FMM in $O(n^\omega)$ it is open whether reachable-APSD enumeration admits for a non-combinatorial algorithm with delay in $o(\maxdeg)$.

\section{Sorted-APSD Enumeration}
\label{sec:s-APSD}
Although the presented algorithms for SSSD enumeration happen to already produce output sorted by distance, running these algorithms in sequence as presented in \autoref{sec:APSD} does not preserve this output property and thus does not solve sorted-APSD enumeration.

As first observation, we note that a delay in $O(\avgdeg)$ (as we have for unweighted unconstrained APSD enumeration) would be surprising, as the conditional lower bound of \autoref{thm:r-APSD-lower} for reachable-APSD enumeration transfers to the sorted-APSD enumeration problem:
By enumerating the solutions to APSD sorted by increasing distance and stopping the enumeration as soon as the first infinite distance is emitted, one solves reachable-APSD with a delay at least as good as sorted-APSD.

\begin{corollary}
	\label{cor:s-APSD-lower}
	No algorithm can solve sorted-APSD enumeration with delay in $O(\avgdeg)$ and preprocessing in $O(n + m)$ unless $d \times d$ BMM can be solved in $O(d^2)$ total time.
\end{corollary}

One can solve sorted-APSD enumeration by running the respective SSSD algorithms in parallel.
In the unweighted case, the individual breadth-first searches are run in a cyclic fashion, each search being interrupted as soon as the reported next lowest hop distance increases.

\begin{restatable}{theorem}{thmUSApsdUpper}
	\label{thm:u-s-APSD-upper}
	Unweighted sorted-APSD enumeration can be solved with delay in $O(\maxdeg)$, with $\Theta(n^2)$ lazy-initialized memory and space complexity in $\Theta(n^2)$.
\end{restatable}
\begin{proof}
	Similar to \autoref{thm:u-APSD-upper}, the algorithm maintains a solution queue $Q$ and begins by computing the trivial $n$ solutions with distance $0$ along with the maximum degree $\maxdeg$; emitting the first $\frac{n}{2}$ solutions in constant delay and the second half with delay in $O(\maxdeg)$.
	Afterwards it initializes on every vertex $v$ a BFS $B_v$ that maintains its private solution queue $Q_v$; each BFS $B_v$ is initially executed up to the point that it would emit the solution $(v,v,0)$ from this queue.
	These solutions are discarded instead, as they have already been added to $Q$.

	The algorithm keeps track of the next expected hop distance $d$, which is initialized to $d = 1$.
	All BFS instances are put in a queue $B$, which is processed as follows:
	As long as $B$ is not empty, select the head entry $B_s$.
	If $Q_s$ has an entry, move it to the end of $Q$ and set $d$ to the distance of this solution.
	Then, advance $B_s$ to produce its next solution $(s,t,\dist(s,t))$ and put it in its private solution queue $Q_s$.
	If the distance is $\dist(s,t) = d$, then do nothing.
	Should $\dist(s,t)$ be larger, then move $B_s$ to the end of $B$.
	If $B_s$ is finished and thus has not produced a new solution, remove it from $B$.
	Afterwards, the loop over $B$ continues.

	Note that each BFS $B_s$ is effectively executed until it emits its first solution of the next higher hop distance; this respective solution is held back until all other searches produced all remaining solutions with the current hop distance and $B_s$ arrives again at the head of $B$.
	Thus, the algorithm enumerates all solutions sorted by increasing distance.

	We use accounting to prove that this procedure indeed emits solutions from $Q$ with delay in $O(\maxdeg)$.
	The initial $n$ solutions are produced with constant delay, the first half being also emitted with constant delay. These solutions are used to pay for the degree computation.
	The second set of $\frac{n}{2}$ solutions is also produced with constant delay but is additionally charged $4 \cdot \maxdeg$ units, resulting in $2n\maxdeg$ credit.
	From that, $n\maxdeg$ credit is used to pay for the initialization of $n$ BFS instances, the queue $B$ and the distance tracker $d$.

	Recall from \autoref{thm:u-SSSD-upper} and \autoref{cor:c-sssd} that unweighted SSSD can be solved with delay in $O(\maxdeg)$ with an algorithm that produces the solutions sorted by increasing distance.
	Thus, each loop iteration costs $\maxdeg$ credit.
	Each respective first time a BFS $B_s$ is the head entry of $B$, its private solution queue $Q_s$ will be empty.
	These $n$ iterations are paid for with the remaining $n\maxdeg$ credit from the initial solutions.
	For all other iterations, the solution being moved to $Q$ is charged $\maxdeg$ which pays for all computation in this iteration.

	As each solution is assigned a cost in $O(\maxdeg)$, the algorithm solves unweighted sorted-APSD enumeration with delay in $O(\maxdeg)$.
\end{proof}

Without trivial self-distances, the first solution to be emitted corresponds to an arbitrary edge.
As it takes up to $\Omega(n)$ time to find the first edge, this results in a lower bound on the delay or preprocessing of algorithms for sorted-no-self-APSD.

\begin{restatable}{theorem}{thmUSApsdNoselfLower}
	\label{thm:u-s-APSD-noself-lower}
	No algorithm can solve unweighted sorted-no-self-APSD enumeration with both delay and preprocessing in $o(n)$.
\end{restatable}
\begin{proof}
	An adversary can construct a graph in which two vertices are connected by an edge and the remaining $n-2$ vertices are isolated.
	For the first $n-2$ queries of the algorithm, the adversary only returns isolated vertices.
	Should the algorithm emit a solution $(s,t,1)$ before having made at least $n-2$ queries, the adversary can claim for at least one of $s$ or $t$, that this vertex is isolated and thus must not be part of any solution.
	If the algorithm emits a solution with distance greater or smaller $1$, it violates the sorted- or the no-self-requirement.
	Therefore, the algorithm has to make at least $\Omega(n)$ queries to the adversary before it can safely produce the first output.
\end{proof}

This bound is tight in the sense that with $O(n)$ preprocessing one can achieve the same delay as for sorted-APSD \emph{with} self-distances.
Basically the preprocessing time is used to locate all edges; afterwards producing and emitting solutions works roughly as for \autoref{thm:u-s-APSD-upper}.

\begin{restatable}{theorem}{thmUSApsdNoselfUpper}
	\label{thm:u-s-APSD-noself-upper}
	With $O(n)$ preprocessing, unweighted sorted-no-self-APSD enumeration can be solved with delay in $O(\maxdeg)$, with $\Theta(n^2)$ lazy-initialized memory and space complexity in $\Theta(n^2)$.
\end{restatable}
\begin{proof}
	Within $O(n)$ preprocessing time the algorithm can divide the vertices into those with (out\nobreakdash-)degree 0 and those with at least one outgoing edge; additionally it computes $\maxdeg$.
	Afterwards it transforms the edges into solutions with hop distance 1 within constant time per solution, but charges additional $4 \cdot \maxdeg$~units for each.
	Of this credit, $\min(m,n)\cdot \maxdeg$~units are used to initialize BFS instances on all vertices with at least one outgoing edge; $3m\Delta$~units pay for the at most $\min(m,n)+2m$ iterations in which these instances produce (unwanted or duplicate) solutions with distance $\leq 1$.
	When all BFS instances are fully finished, the algorithm emits all $\infty$-distances of the vertices with (out-)degree 0 within constant time per solution.
	The rest of the analysis remains unchanged, therefore the algorithm solves unweighted sorted-no-self-APSD enumeration with $O(n)$ preprocessing and delay in~$O(\maxdeg)$.
\end{proof}

If the graph has non-negative edge weights, we again use Dijkstra's algorithm instead of BFS to compute the distances.
Running $n$ instances of Dijkstra's algorithm in a cyclic fashion (as done with BFS in the unweighted case) would not guarantee sorted output; instead a priority queue is used to choose the respective Dijkstra instance to be advanced next.

\begin{restatable}{theorem}{thmWSApsdUpper}
	\label{thm:w-s-APSD-upper}
	Sorted-APSD enumeration on graphs with edge weights in $\{0, \ldots, n^c\}$, for some constant $c$, can be solved with delay in $O(\maxdeg + \log(n))$, with $\Theta(n^2)$ lazy-initialized memory and space complexity in $\Theta(n^2)$.
\end{restatable}
\begin{proof}
	To derive this result, the algorithm and analysis from \autoref{thm:u-s-APSD-upper} are adjusted:
	The BFS instances are replaced with Dijkstra's algorithm using Strict Fibonacci Heaps as data structure, increasing the time for both the initialization and each loop iteration to $O(\maxdeg + \log(n))$.
	Accordingly, every $\maxdeg$ units paid by solutions before are replaced with $\maxdeg + \log(n)$ units, thereby again covering all computation costs.

	We replace the queue of BFS instances with a priority queue to select in each loop iteration the respective instance of Dijkstra to be advanced.
	The initial key of each priority queue entry is the minimum edge weight adjacent to the start vertex of this Dijkstra instance.
	In each loop iteration, the current instance is determined by \textsc{extractMin}.
	If the instance produces a new solution, the instance is re-inserted into the priority queue with the distance of the new solution as key.
	Therefore, the key of each instance always represents the next solution this instance will contribute to the overall solution queue.
	Repeatedly selecting the minimum among all keys ensures that the algorithm enumerates the solutions sorted by increasing distance.
	By using a priority queue with both \textsc{extractMin} and \textsc{insert} in $O(\log(n))$, those operations are dominated by the existing computation time per loop iteration.

	As each solution is assigned a cost in $O(\maxdeg + \log(n))$, the algorithm solves sorted-APSD enumeration with non-negative edge weights with delay in $O(\maxdeg + \log(n))$.
\end{proof}

Without trivial self-distances, the time until an algorithm can safely report the first solution increases to $\Omega(m + n)$, as this solution corresponds to the edge with minimum weight.

\begin{restatable}{theorem}{thmWSApsdNoselfLower}
	\label{thm:w-s-APSD-noself-lower}
	No algorithm can solve sorted-no-self-APSD enumeration on graphs with edge weights in $\{0, \ldots, n^c\}$, for some constant $c$, with both delay and preprocessing in $o(m + n)$.
\end{restatable}
\begin{proof}
	An adversary can construct an arbitrary graph with $n$ vertices and $m$ edges.
	If there are isolated vertices, the adversary presents those first before revealing the first edge to the algorithm.
	For the first $m - 1$ neighborhood queries of the algorithm that result in an actual edge, the adversary returns $3$ as edge weight.
	Thus, should the algorithm emit a solution $(s,t,d)$ before having made at least $m + n$ queries, the adversary is free to choose the remaining edge weight as either $1$ or $2$ and thereby falsify the algorithm's output:
	If $d \geq 2$, the adversary chooses edge weight $1$ and the algorithm violates the sorted-requirement.
	For $d = 1$, the adversary sets the last edge weight to $2$ and thus $s$ and $t$ have at least distance $2$.
	Therefore, the algorithm has to make at least $\Omega(m + n)$ queries to the adversary before it can safely produce the first solution.
\end{proof}

Again we can match the delay upper bound for sorted-APSD while only using the minimum amount of preprocessing, as $O(m + n)$ preprocessing time is enough for the algorithm to perform all initialization steps and afterwards run the loop of Dijkstra instances unchanged.

\begin{restatable}{theorem}{thmWSApsdNoselfUpper}
	\label{thm:w-s-APSD-noself-upper}
	With $O(m + n)$ preprocessing, sorted-no-self-APSD enumeration on graphs with edge weights in $\{0, \ldots, n^c\}$, for some constant $c$, can be solved with delay in $O(\maxdeg + \log(n))$, with $\Theta(n^2)$ lazy-initialized memory and space complexity in $\Theta(n^2)$.
\end{restatable}
\begin{proof}
	We adapt the algorithm from \autoref{thm:w-s-APSD-upper} by moving all initialization to the $O(m + n)$ preprocessing phase while skipping self-distance solutions.
	Afterwards, the loop of Dijkstra instances runs unchanged and produces solutions with delay in $O(\maxdeg + \log(n))$.
\end{proof}

\section{A Note on Undirected Graphs}
\label{sec:undirected-APSD}

When applied to an undirected graph, an APSD enumeration algorithm would emit the distance between each unordered pair of vertices $u \neq v$ twice.
We can however always adjust the algorithm to instead, for all vertices $u,v$, emit exactly one of $(u,v,\dist(u,v))$ and $(v,u,\dist(v,u))$, without compromising the asymptotic delay:
The adjusted algorithm only emits the solution found first and skips the second.
By emitting all solutions with twice the delay, each first solution pays for the omission of the respective second occurrence.

Whilst such tie-breaking can easily always be done with $\Theta(n^2)$ lazy-initialized memory in which the produced solutions are tracked, all algorithms presented in this paper allow for a constant-time tie-breaking without additional memory.
This is essentially due to the fact that our APSD enumeration algorithms can control the order in which the SSSD enumeration algorithms they build on are executed:
\begin{itemize}
	\item For the unconstrained algorithms presented in \autoref{sec:APSD} and the algorithms for reachable-APSD enumeration of \autoref{cor:r-apsd}, selecting the start vertices for the SSSD algorithms with increasing vertex id and only emitting solutions $(u,v,\cdot)$ with $u \leq v$ suffices, as apart from the trivial self-distances all solutions are produced by the SSSD algorithms.
	\item The algorithm of \autoref{thm:u-APSD-noself-upper} for unweighted no-self-APSD enumeration can be adjusted the same way; additionally the initial pass over all vertices and edges has to respect the same vertex order.
	\item Weighted no-self APSD enumeration according to \autoref{thm:w-APSD-noself-upper} changes the vertex order according to their node degree.
	As BucketSort is stable, all following operations iterate over vertices sorted by node degree first and vertex id second.
	If tie-breaking uses the same ordering, duplicate solutions can again be identified in constant time.
	\item The algorithms for sorted APSD enumeration in \autoref{sec:s-APSD} executes SSSD algorithms in parallel instead of consecutively, but can again apply tie-breaking by vertex id to select the respective SSSD algorithm instance to advance next.
\end{itemize}

\section{Open Problems}
\label{sec:conclusion}
In this paper we only started investigating what an enumeration point of view can reveal about the problems SSSD and APSD.
Improvements on our results, different model assumptions, as well as other input types remain wide open for future research.

In \autoref{cor:fast-w-APSD-upper} we explained how to transfer improvements of SSSD algorithms for weighted graphs to unconstrained-APSD enumeration.
It is however open, whether such improvements could be applied to speed up the delay for row-wise-APSD/SSSD enumeration on weighted graphs.
With the lower bound of \autoref{thm:SSSD-lower} in mind this leaves the question, whether one can get rid of the additional $O(\log(n))$ steps per delay.

For weighted no-self-APSD enumeration we showed a lower bound of $\Omega(\Delta)$ on either preprocessing or delay in \autoref{thm:w-APSD-noself-lower}, but only manage to solve the problem with preprocessing in $O(n)$ in \autoref{thm:w-APSD-noself-upper}.
Closing this gap remains an interesting open problem.

If one lifts the restriction that input-data is read-only, there might be room for improvements in the space complexity, as Kammer and Sajenko recently showed how to solve BFS in-place in linear time \cite{kammerLinearTimeInPlaceDFS2019a}.

At last, we have not at all addressed SSSD or APSD for graphs with negative edge weights.
Observe that in this general case, even emitting the first distance answers the question if a strongly connected graph contains a negative cycle.
Whilst this problem seems to be easier than computing distances, state of the art algorithms require about as much time to decide the presence of negative cycles as solving SSSD or even APSD completely (see e.\,g. the prior work section in \cite{bringmannImprovedAlgorithmsComputing2017}).



\bibliography{shortest-distances-as-enumeration-problem}
\end{document}